\documentclass[11pt]{article}
\usepackage[T1]{fontenc}
\usepackage[cp1250]{inputenc}
\usepackage{amsmath}
\usepackage{amsthm}
\usepackage{amsfonts}
\usepackage{graphicx}
\usepackage{yfonts}
\usepackage{color}
\usepackage[normalem]{ulem}
\usepackage{bm}
\usepackage{bbm}
\usepackage{cite}
\newcommand{\CPn}{{\mathbb{C}P^{n-1}}}

\newcommand{\Hil}{{\mathcal{H}}}
\DeclareMathOperator{\tr}{tr}
\DeclareMathOperator{\imag}{i}

\DeclareMathOperator{\e}{e}
\DeclareMathOperator{\der}{d}

\DeclareMathOperator{\Real}{Re}
\DeclareMathOperator{\Imag}{Im}
\def\<{\langle}
\def\>{\rangle}

\newtheorem*{lemat*}{Lemma}
\def\oper{{\mathchoice{\rm 1\mskip-4mu l}{\rm 1\mskip-4mu l}
{\rm 1\mskip-4.5mu l}{\rm 1\mskip-5mu l}}}

\title{\bf Non-Markovianity of geometrical qudit decoherence}
\author{ Katarzyna Siudzi{\'n}ska\\ {\it Institute of Physics, Faculty of Physics, Astronomy and Informatics}\\ {\it Nicolaus Copernicus University} \\ {\it ul. Grudzi\k{a}dzka 5/7, 87-100 Toru{\'n}, Poland}\\
kasias@doktorant.umk.pl}

\begin{document}

\maketitle

\begin{abstract}
In the following paper, we generalize the geometrical framework of qubit decoherence to higher dimensions. The quantum mixed state is represented by the probability distribution, which is the K\"ahler function on the projective Hilbert space. The Markovian master equation for density operators turns out to be equivalent to the Fokker-Planck equation for quantum probability distributions. Several examples are analyzed, featuring different generalizations of the Pauli channel.
\end{abstract}

%
%

\section{Introduction}

Geometrization of quantum mechanics has been given a lot of attention lately. While the description of the Hamiltonian dynamics is well developed for both pure and mixed states, the dynamics of open quantum system has not been given enough notice. Therefore, our goal is to apply the geometrical structures of quantum mechanics to the problem of non-Markovian evolution and decoherence processes of open quantum systems.

First, let us recall some basic concepts of the theory of open quantum systems. It is well known that every quantum evolution can be described by a completely positive, trace-preserving map $\Lambda(t):\mathcal{B}(\mathcal{H})\to\mathcal{B}(\mathcal{H})$ with $\Lambda(0)=\oper$. In other words, the mapping takes the initial state $\rho(0)$ into the evolved state at time $t$, i.e. $\Lambda(t)[\rho(0)]=\rho(t)$. The Markovianity property of a dynamical map is determined by its divisibility \cite{RHP, Wolf-Isert}. Namely, the evolution is Markovian iff $\Lambda(t)$ is CP-divisible -- that is, iff it can be written in the following form,
\begin{equation}
  \Lambda(t) = V(t,s) \Lambda(s),
\end{equation}
where $V(t,s)$ is completely positive and trace-preserving (CPT) for all $t\geq s$. This property translates into the very specific form of the time-local generator $\mathcal{L}(t)$ that appears in the master equation,
\begin{equation}\label{master_lambda}
  \frac{\der}{\der t} \Lambda(t)[\rho] =\mathcal{L}(t) \Lambda(t)[\rho].
\end{equation}
This form is called the time-dependent Gorini-Kossakowski-Sudarshan-Lindblad (GKSL) form \cite{GKS, L} and it reads,
\begin{equation}\label{generator_general}
\mathcal{L}(t)[\rho] =-i[H(t),\rho] + \frac 12 \sum_{k=0}^{n^2-1} \gamma_k(t)\left( V_k(t)\rho V_k^\dagger(t) -\frac 12 [V_k(t)V_k^\dagger(t),\rho]_+\right) ,
\end{equation}
where $H(t)$ is the Hamiltonian, $V_k$ -- the noise operators, and the decoherence rates $\gamma_k(t)\geq 0$.
Iff $V(t,s)$ is positive (but not necessarily CPT) and invertible, then $\Lambda(t)$ satisfies the weaker condition for Markovianity proposed in \cite{PRL-Sabrina},
\begin{equation}\label{P-div}
  \frac{\der}{\der t} ||\Lambda(t)[X]||_1 \leq 0
\end{equation}
for every Hermitian $X$, with $||X||_1:=\tr\sqrt{X^\dagger X}$ being the trace norm.

In Section 2, we introduce the geometrical language to describe quantum mechanics on the K\"{a}hler manifolds. Section 3 lists the most important results for general qudit decoherence. Sections 4-7 deal with four different generalizations of the Pauli channel to higher dimensions. There, we analyze the properties of several CPT maps (which, in general, correspond to different dynamics) and their time-dependent generators. In Section 8, we examine the conditions for divisibility in the geometric approach. Final conclusions are gathered in Section 9.

\section{Geometrical formulation of quantum mechanics}

The geometrical formulation of quantum mechanics recognizes the projective Hilbert space $\mathbb{P}\Hil$ as the space of quantum states  \cite{Kibble,Cirelli,Cantoni,Anandan} (for recent reviews, see \cite{Brody, Ashtekar,Moretti}). For every point in $\mathbb{P}\Hil$ there exists the corresponding rank-1 projector $|\psi\>\<\psi|$ and a ray in the Hilbert space $\Hil$ passing through $\psi$ (see also \cite{Jamiolkowski, Bengtsson, MARMO, Heydari}). If we choose $\Hil = \mathbb{C}^n$, then the space of states
\begin{equation}
\mathbb{P}\mathcal{H} = \CPn  = \mathrm{U}(n)/\mathrm{U}(n-1)
\end{equation}
is the $(n-1)$-dimensional complex space equipped with the Fubini-Study metric $g$ and the symplectic form $\omega$ such that the K\"ahler form
\begin{equation}
\mathcal{K} = g + \imag\omega.
\end{equation}
The triple $(\CPn, g, \omega)$ is the K\"ahler manifold \cite{Heydari-JPA}. On this space, one can define the K\"ahler functions \cite{GNS} in the following way. 
A function $f : \CPn \to \mathbb{C}$ is K\"ahlerian if and only if its Hamiltonian vector field $X_f$, which is given by the equation $\der f = \omega(X_f,\, \cdot\, )$,
is a Killing vector field -- that is, if and only if the Lie derivative $\mathfrak{L}_{X_f} g =0$. These functions form a linear subspace in the space of all functions
$\mathcal{F}(\CPn) :=\{ f : \CPn \to \mathbb{C} \}$. 

With every operator $A \in \mathcal{B}(\mathcal{H})$, one can associate the function $f_A : \CPn \to \mathbb{C}$ given by
\begin{equation}\label{fF}
  f_A([\psi]) := \frac{\<\psi|A|\psi\>}{\<\psi|\psi\>} .
\end{equation}
The function $f_A$ is simply the expectation value of the corresponding operator $A$ in a normalized state, and therefore we call it the expectation value function.
Note that if we take the operators $A,B \in \mathcal{B}(\mathcal{H})$, $C=-2\imag[A,B]$, and $D=2[A,B]_+:=2(AB+BA)$, then the corresponding K\"ahler functions $f_A,f_B,f_C,f_D$ are connected to each other by
\begin{equation}
  f_C = \{f_A,f_B\} = -X_{f_A}(f_B),\qquad f_D = \{f_A,f_B\}_+.
\end{equation}
Here, $\{f_A,f_B\} := \omega(X_{f_A},X_{f_B})$ is the Poisson bracket on $\CPn$ defined by the symplectic form $\omega$, and $\{f_A,f_B\}_+ := g(X_{f_A},X_{f_B})+4f_Af_B$ is the symmetric bracket \cite{Ashtekar} given by the Fubini-Study metric $g$.

Let us introduce $(\psi_0,\psi_1,\dots,\psi_{n-1})\in\mathbb{C}^n$ in terms of the octant coordinate system,
\begin{equation}
(\psi_0,\ \psi_1,\ \dots\ ,\ \psi_{n-1})=(N_0,\ N_1\e^{\imag\nu_1},\ \dots\ ,\ N_{n-1}\e^{\imag\nu_{n-1}}),
\end{equation}
where $0\leq\nu_i\leq 2\pi$ and $\sum_{i=0}^{n-1}N_i=1$. In local coordinates, we set
\begin{equation}
\begin{cases}
N_0=\cos\theta_1\sin\theta_2\sin\theta_3\ \dots\ \sin\theta_{n-1},\\
N_1=\sin\theta_1\sin\theta_2\sin\theta_3\ \dots\ \sin\theta_{n-1},\\
N_2=\cos\theta_2\sin\theta_3\ \dots\ \sin\theta_{n-1},\\
\vdots\\
N_{n-1}=\cos\theta_{n-1}
\end{cases}
\end{equation}
with $0\leq\theta_i\leq\pi/2$. Now, the Fubini-Study metric $g$ and the symplectic form $\omega$ read
\begin{equation}\label{omega}
\begin{split}
g&=\der N_0^2+\sum_{i=1}^{n-1}\Big[\der N_i^2+N_i^2(1-N_i^2)\der\nu_i^2-2\sum_{j=i+1}^{n-1}N_i^2N_j^2\der\nu_i\der\nu_j\Big],\\
\omega&=\sum_{i=1}^{n-1}N_i\der N_i\wedge\der\nu_i.
\end{split}
\end{equation}
The symplectic form defines a volume element on $\CPn$,
\begin{equation}\label{}
  {\rm Vol}(\CPn) = \int_{\CPn} \omega=\frac{\pi^{n-1}}{(n-1)!}.
\end{equation}
Observe that from the geometrical point of view $N_k$'s form the positive hyperoctant of an $(n-1)$-sphere, whereas the phases $\nu_k$'s form an $(n-1)$-torus \cite{Bengtsson}.

Now, for a given density operator $\rho$, let us introduce the following function,
\begin{equation}\label{}
p([\psi]):=\frac{(n-1)!}{\pi^{n-1}}\frac{\<\psi|\rho|\psi\>}{\<\psi|\psi\>}.
\end{equation}
From definition, $p([\psi])$ is a probability distribution on $\CPn$, which means that $p([\psi]) \geq 0$ and
\begin{equation}
  \int_{\CPn} p([\psi]) \omega = 1 \ ,
\end{equation}
where $\omega$ is given by (\ref{omega}). This function corresponds to a legitimate density operator iff $p$ is a K\"ahler function. It is worth noting that $p$ describes a pure state $\rho=|\psi\>\<\psi|$ iff $p([\psi])={(n-1)!}/{\pi^{n-1}}$.

\section{General qudit evolution}

From now on, we will be interested only in the non-Hamiltonian evolution of a quantum system, therefore limiting our discussion to the following family of generators,
\begin{equation}\label{Lindblad}
\mathcal{L}(t)[\rho]=\frac 12 \sum_{k=0}^{n^2-1}\gamma_k(t)\left(V_k(t)\rho V_k^\dagger(t)
-\frac 12 [V_k^\dagger V_k,\rho]_+\right).
\end{equation}
The above formula can be rewritten with the use of double bracket structures,
\begin{equation}
\begin{split}
\mathcal{L}(t)[\rho]=-\frac 18 \sum_{k=1}^{n^2-1}\gamma_k(t)\Bigg(&[V_k^\dagger(t),[V_k(t),\rho]] +[V_k(t),[V_k^\dagger(t),\rho]]\\&+ [V_k^\dagger(t),[V_k(t),\rho]_+] -[V_k(t),[V_k^\dagger(t),\rho]_+]\Bigg).
\end{split}
\end{equation}
In the geometrical language, this is equivalent to the following equation for the probability distribution,
\begin{equation}\label{FP}
l(t)[p]=\frac{1}{16}\sum_{k=1}^{n^2-1}\gamma_k(t)\Bigg(|X_{v_k(t)}|^2p +\Imag X_{v_k(t)}\{v_k^*(t),p\}_+ \Bigg),
\end{equation}
where the master equation (\ref{master_lambda}) translates into
\begin{equation}\label{ME_geom}
\frac{\der}{\der t}p(t)=l(t)[p(t)].
\end{equation}
Here, $p(t)=\<\rho(t)\>$, $v_k(t)=\<V_k(t)\>$, where $\<A\>$ denotes the expectation value of the operator $A$, and $v_k^*(t)$ is the complex conjugation of $v_k(t)$. The above formula is clearly the Fokker-Planck equation \cite{Risken} for the probability distributions of quantum states. Eq.~(\ref{FP}) can be naturally divided into the quantum-classical (QC) part,
\begin{equation}\label{QC}
l^{\mathrm{QC}}(t)[p]=\frac{1}{16}\sum_{k=1}^{n^2-1}\gamma_k(t)|X_{v_k(t)}|^2p,
\end{equation}
and the purely quantum (PQ) part,
\begin{equation}\label{PQ}
l^{\mathrm{PQ}}(t)[p]=\frac{1}{16}\sum_{k=1}^{n^2-1}\gamma_k(t)\Imag X_{v_k(t)}\{v_k^*(t),p\}_+.
\end{equation}
This distinction will become more clear if we observe that the Poisson bracket is the only bracket that transforms $k$-poles into $k$-poles, where by the $k$-pole one understands the harmonic function of degree $k$ in the multipole expansion. For $n=2$, the multipole functions are simply the spherical harmonics. Therefore, replacing the quantum distribution $p(t)$, composed of a monopole and $n-1$ dipoles, with a classical one, whose expansion consists of all possible multipoles, will result in a valid master equation only for $l^{PQ}(t)[p]=0$. Therefore, whenever the r.h.s. of any geometrical master equation has at least one symmetric bracket in it, the equation describes quantum dynamics. It turns out that one can choose the noise operators in such a way that the purely quantum part (\ref{PQ}) vanishes, and the classical-quantum part (\ref{QC}) does not vanish for non-zero noise operators.

Quantum evolution can be described with the use of the CPT map $\Lambda_t$. Let us assume that this map has the following Kraus decomposition,
\begin{equation}\label{Kraus}
\Lambda(t)[\rho(0)]=\sum_{k=0}^{n^2-1}\pi_k(t)A_k(t)\rho(0)A_k^\dagger(t),
\end{equation}
where $p_k\geq 0$, $\sum_{k=0}^{n^2-1}\pi_k(t)A_k^\dagger(t)A_k(t)=\oper$, $\pi_0(0)=1$.
It is possible to rewrite equation (\ref{Kraus}) in the form which is very similar to the r.h.s. of the GKSL form; namely,
\begin{equation}
(\Lambda(t)-\Lambda(0))[\rho(0)]=\sum_{k=0}^{n^2-1}\pi_k(t)\left(A_k(t)\rho(0)A_k^\dagger(t)-\frac 12 [A_k^\dagger(t) A_k(t),\rho(0)]_+\right).
\end{equation}
In the geometrical formulation, this corresponds to
\begin{equation}
p(t)-p(0)=\frac{1}{8}\sum_{k=1}^{n^2-1}\pi_k(t)\Bigg(|X_{a_k(t)}|^2p(0) +\Imag X_{a_k(t)}\{a_k^*(t),p(0)\}_+\Bigg)
\end{equation}
with $p(t)=\<\rho(t)\>$ and $a_k(t)=\<A_k(t)\>$.

In the previous work \cite{moje}, the authors analyzed the properties of the random unitary qubit evolution. Now, we would like to generalize this picture to higher dimensions. However, it turns out that there are at least four natural generalizations of the Pauli matrices:

\begin{enumerate}
\item {\it the Gell-Mann matrices}, which are Hermitian but non-unitary;
\item {\it the Weyl operators}, which are unitary but non-Hermitian;
\item {\it the tensor products of the Pauli matrices}, which are Hermitian and unitary, but they are applicable only for $n=2^r$;
\item {\it the projectors on mutually unbiased bases}, which are again both Hermitian and unitary, but they are applicable only for $n=s^r$ with prime $s$.
\end{enumerate}

The following sections are dedicated to the analysis of time-local generators where the Pauli matrices were replaced with the generalized operators from one of the abovementioned sets.

\section{Gell-Mann matrices}

When choosing the generalization in which the Hermiticity of the noise operators is preserved, the formula (\ref{Lindblad}) can be rewritten with the use of a double commutator structure,
\begin{equation}\label{L_GM}
\mathcal{L}(t)[\rho]=-\frac 14\sum_{k_1,k_2=0}^{n-1}\gamma_{k_1k_2}(t)[\tau_{k_1k_2},[\tau_{k_1k_2},\rho]],
\end{equation}
where $\tau_{k_1k_2}$ are the Gell-Mann matrices defined as follows \cite{qudits}:
\begin{equation}
\begin{aligned}
&\tau_{k_1k_2}^S=E_{k_1k_2}+E_{k_2k_1},&\qquad 0\leq k_1<k_2\leq n-1,\\
&\tau_{k_1k_2}^A=-\imag(E_{k_1k_2}-E_{k_2k_1}),&\qquad 0\leq k_1<k_2\leq n-1,\\
&\tau_{k_1k_1}^D=\sqrt{\frac{2}{k_1(k_1+1)}}\left(\sum_{j=0}^{k_1-1}E_{jj}-k_1E_{k_1k_1}\right),&\qquad 1\leq k_1\leq n-1,\\
&\tau_{00}^D=\sum_{j=0}^{n-1}E_{jj},&
\end{aligned}
\end{equation}
with $E_{ij}$ being the $n\times n$ matrix with 1 as its $(i,j)$'th entry and 0 elsewhere.
In the language of quantum probability distributions, the master equation (\ref{L_GM}) reads
\begin{equation}
l(t)[p(t)]=\frac{1}{16}\sum_{k_1,k_2=0}^{n-1}\gamma_{k_1k_2}(t)\{f_{k_1k_2},\{f_{k_1k_2},p(t)\}\} =\frac{1}{4} \Delta_{\gamma(t)} p(t),
\end{equation}
with $f_{k_1k_2}=\<\tau_{k_1k_2}\>$ and 
\begin{equation}
\Delta_{\gamma(t)}:=\frac{1}{4}\sum_{k_1,k_2=0}^{n-1}\gamma_{k_1k_2}(t)X_{f_{k_1k_2}}^2
\end{equation}
being the generalized Laplacian on $\mathbb{C}P^{n-1}$.

In general, the eigenfunctions of $l(t)$ are time-dependent. This makes finding the eigenvalue formulas highly problematic. However, if we narrow down our interest to the case in which $\gamma_{k_1k_2}^S(t)+\gamma_{k_1k_2}^A(t)=:\gamma(t)$, then it turns out that $f_{k_1k_2}$'s are the desired eigenfunctions, and therefore
\begin{equation}
l(t)[f_{k_1k_2}]=l_{k_1k_2}(t)f_{k_1k_2}.
\end{equation}
The corresponding eigenvalues are listed below:
\begin{equation}\label{eigenbasis_GM_special}
\begin{aligned}
l_{k_1k_2}^{S/A}=&-\gamma_{k_1k_2}^{A/S}-\frac{k_1}{2(k_1+1)}\gamma_{k_1k_1}^D -\sum_{j=k_1+1}^{k_2-1}\frac{\gamma_{jj}^D}{2j(j+1)}\\&-\frac{k_2+1}{2k_2}\gamma_{k_2k_2}^D-\frac{\gamma}{2}(n-2),
\qquad\qquad\qquad 0\leq k_1<k_2\leq n-1,\\
l_{k_1k_1}^D=&-\frac{\gamma}{2}n,\qquad\qquad\qquad\qquad\qquad\qquad\qquad\quad  1\leq k_1\leq n-1,
\end{aligned}
\end{equation}
and $l_{00}^D=0$.

In the isotropic case, when $\gamma_{k_1k_2}(t)=:\gamma(t)$, equation (\ref{L_GM}) simplifies to
\begin{equation}\label{lap}
\dot{p}(t)=\frac{\gamma}{16}\sum_{k_1,k_2=0}^{n-1}X_{f_{k_1k_2}}^2p(t) =\frac{\gamma}{4}\Delta p(t),
\end{equation}
with $\Delta$ being the Laplacian. This corresponds to the dephasing channel for a qudit. Moreover, every isotropic evolution whose noise operators are given by a unitary rotation of the Gell-Mann matrices leads to the same dynamical equation for $p(t)$. To check the validity of (\ref{lap}), one needs the following lemma.

\begin{lemat*}
Every dipole K\"ahler function $f_{k_1k_2}$, i.e. $f_{k_1k_2}$ for $k_1\neq 0$ and $k_2\neq 0$, is an eigenfunction of the Laplace operator to the same eigenvalue,
\begin{equation}
\Delta f_{k_1k_2}=-4nf_{k_1k_2}.
\end{equation}
\end{lemat*}

\begin{proof}
Let us start from the isotropic ($\gamma_k=1$) algebraical master equation (\ref{L_GM}) for a traceless Hermitian operator $\rho^\prime$,
\begin{equation}
\dot{\rho^\prime}=-\frac 14\sum_{k=1}^{n^2-1}[\tau_k,[\tau_k,\rho^\prime]],
\end{equation}
where $\tau_k$ are the Gell-Mann matrices and
\begin{equation}
\rho^\prime=\sum_{k=1}^{n^2-1}x_k\tau_k.
\end{equation}
Using the properties of $\tau_k$ (see e.g. \cite{Bengtsson}),
\begin{equation}
[\tau_i,\tau_j]=2\imag\sum_{k=0}^{n^2-1}\epsilon_{ijk}\tau_k,
\end{equation}
\begin{equation}
\sum_{i,j=0}^{n^2-1}\epsilon_{ijk}\epsilon_{ijl}=n\delta_{kl},
\end{equation}
we arrive at
\begin{equation}
\dot{\rho^\prime}=-\frac 14\sum_{k,l=1}^{n^2-1}x_l[\tau_k,[\tau_k,\tau_l]]= -\sum_{i,j,k,l=1}^{n^2-1}x_l\tau_i\epsilon_{jkl}\epsilon_{jki}=-n\rho^\prime.
\end{equation}
In the geometrical framework, the above equation is equivalent to
\begin{equation}
\dot{p^\prime}=\frac 14\Delta p^\prime=-np^\prime,
\end{equation}
with $p^\prime$ being the expectation value function of $\rho^\prime$.
Therefore, the eigenvalue equation for the Laplace operator reads
\begin{equation}
\Delta f_k=-4nf_k
\end{equation}
for $k\neq 0$.
\end{proof}

\section{Weyl operators}

When we are interested in generalizing the two-dimensional case in a way that the unitarity of the noise operators is preserved, a natural choice is the following form of the generator,
\begin{equation}\label{L_Weyl}
\mathcal{L}(t)[\rho]=\frac 12\sum_{k_1,k_2=0}^{n-1}\gamma_{k_1k_2}(t)\Big(U_{k_1k_2}\rho U_{k_1k_2}^\dagger-\rho\Big).
\end{equation}
In the case of the master equation with the generator of evolution given by (\ref{L_Weyl}), the Kraus form of the corresponding channel is well known. Namely, this is the CPT map describing random unitary evolution,
\begin{equation}
\Lambda(t)[\rho]=\sum_{k_1,k_2=0}^n\pi_{k_1k_2}(t)U_{k_1k_2}\rho U_{k_1k_2}^\dagger.
\end{equation}
The newly-introduced $U_{k_1k_2}$ are the Weyl operators, which are defined by (c.f. \cite{Filip2})
\begin{equation}
U_{k_1k_2}:=\sum_{m=0}^{n-k_2-1}\Omega^{k_1m}E_{m,m+k_2},
\end{equation}
with the coefficient $\Omega=\exp\left(\frac{2\pi\imag}{n}\right)$, $k_1$, $k_2=0,\dots,n-1$, and $E_{ml}$ being the matrix with a $1$ in the $(m,l)$'th entry and $0$ elsewhere. 

Equation~(\ref{L_Weyl}) written in the language of probability distributions reads
\begin{equation}
\begin{split}
l(t)[p(t)]=\sum_{k_1,k_2=0}^{n-1}\Bigg( &\frac{\gamma_{k_1k_2}(t)+\gamma_{n-k_1,n-k_2}(t)}{32} |X_{u_{k_1k_2}}|^2p(t) \\&+\frac{\gamma_{k_1k_2}(t)-\gamma_{n-k_1,n-k_2}(t)}{32} \Imag X_{u_{k_1k_2}}\{u_{k_1k_2}^*,p(t)\}_+ \Bigg).
\end{split}
\end{equation}
The eigenvalues $l_{k_1k_2}(t)$ of (\ref{geom_Weyl}) to the eigenfunctions $u_{k_1k_2}$ equal
\begin{equation}
l_{k_1k_2}(t)=\frac 12 \sum_{j_1,j_2=0}^{n-1}\gamma_{k_1k_2}(t)\Bigg(\Real\Omega^{k_2j_1-k_1j_2}-1\Bigg).
\end{equation}
In the special case, where $\gamma_{k_1k_2}(t)=\gamma_{n-k_1,n-k_2}(t)$, the second component under the sum symbol vanishes, and therefore we have
\begin{equation}\label{geom_Weyl}
l(t)[p(t)]=\frac{1}{16}\sum_{k_1,k_2=0}^{n-1}\gamma_{k_1k_2}(t)|X_{u_{k_1k_2}}|^2p(t).
\end{equation}
Note that (\ref{geom_Weyl}) has real eigenvalues.

\section{Tensor products of the Pauli matrices}

It turns out that there is a way to keep the noise operators in higher dimensions both unitary and Hermitian. To make that possible, however, we need to limit our discussion to very specific dimensions of the Hilbert space. Let us write down the following generator,
\begin{equation}\label{L_Pauli}
\mathcal{L}(t)[\rho]=\frac 12\sum_{k_1,\dots,k_N=0}^{3}\gamma_{k_1\dots k_N}(t)\Big(\eta_{k_1\dots k_N}\rho\eta_{k_1\dots k_N}-\rho\Big),
\end{equation}
where the noise operators $\eta_{k_1\dots k_N}$ are tensor products of Pauli matrices $\sigma_{k_j}$,
\begin{equation}
\eta_{k_1\dots k_N}=\bigotimes_{j=1}^{N}\sigma_{k_j},
\end{equation}
and the dimension of the Hilbert space is $n=2^N$.
Similarly to the case of the Gell-Mann matrices, (\ref{L_Pauli}) reduces to a simpler formula with the double commutator structure,
\begin{equation}
\mathcal{L}(t)[\rho]=-\frac 14\sum_{k_1,\dots,k_N=0}^{3}\gamma_{k_1\dots k_N}(t)[\eta_{k_1\dots k_N},[\eta_{k_1\dots k_N},\rho]].
\end{equation}
Moving our interest to the geometrical picture, we arrive at the following -- equivalent -- form for the probability distributions,
\begin{equation}\label{Pauli_L}
l(t)[p(t)]=\frac{1}{16} \sum_{k_1,\dots,k_N=0}^{3}\gamma_{k_1\dots k_N}(t) X_{m_{k_1\dots k_N}}^2p(t)=\frac{1}{4}\Delta_{\gamma(t)} p(t),
\end{equation}
where $m_{k_1\dots k_N}=\<\eta_{k_1\dots k_N}\>$, and the generalized Laplacian $\Delta_{\gamma(t)}:=\frac{1}{4} \sum_{k_1,\dots,k_N=0}^{3}\gamma_{k_1\dots k_N}(t) X_{m_{k_1\dots k_N}}^2$.

Let us remind you that the CPT map satisfying (\ref{L_Pauli}) has the following Kraus representation,
\begin{equation}\label{K_Pauli}
\Lambda(t)[\rho]=\sum_{k_1,\dots,k_N=0}^{3}\pi_{k_1\dots k_N}(t)\eta_{k_1\dots k_N}\rho\eta_{k_1\dots k_N},
\end{equation}
which is equivalent to the following equation for $p(t)$,
\begin{equation}\label{Pauli_K}
p(t)=\Bigg(\frac 12 \Delta_{\pi(t)}+1\Bigg) p(0).
\end{equation}
As the r.h.s. of both (\ref{Pauli_L}) and (\ref{Pauli_K}) includes the Laplacian as the only differential operator, it is fairly easy to combine the two and get the formulas connecting $\pi_{k_1\dots k_N}(t)$ with $\gamma_{k_1\dots k_N}(t)$.
Indeed, substituting (\ref{Pauli_K}) into (\ref{Pauli_L}), from the master equation we obtain the following equations describing the dependence between $\gamma$'s and $\pi$'s,
\begin{equation}
\Bigg[\Delta_{\dot{\pi}(t)}-\frac 12\Delta_{\gamma(t)} \left(\frac 12\Delta_{\pi(t)}+1\right)\Bigg]m_{k_1\dots k_N}=0,
\end{equation}
with
{\medmuskip=1mu
\thinmuskip=1mu
\thickmuskip=1mu
\begin{equation}
\Delta_{\alpha(t)} m_{k_1\dots k_N}=2 \sum_{i_1,\dots,i_N=0}^3\alpha(m_{i_1,\dots,i_N})(t)\prod_{l=1}^N\Bigg[2\left(\delta_{i_lk_l}+\delta_{i_l0}+\delta_{k_l0}-2\delta_{i_l0} \delta_{k_l0}-\frac 12\right)-1\Bigg]m_{k_1\dots k_N}
\end{equation}}
for $\alpha\in\{\gamma,\pi,\dot{\pi}\}$.

\section{Projectors on MUBs}

First, let us recall the basic information on mutually unbiased bases (MUBs). Two orthonormal bases $\{e_k\}$ and $\{f_l\}$ are mutually unbiased iff for every choice of indices
\begin{equation}
|\<e_k|f_l\>|^2=\frac 1n.
\end{equation}
The above condition is satisfied by the eigenbases of the Pauli matrices. In fact, the generator of the Pauli channel can be rewritten in the following form,
\begin{equation}\label{L_gen_pauli}
\mathcal{L}(t)[\rho]=\sum_{\alpha=1}^{3}\gamma_\alpha(t)\Bigg[\sum_{k=0}^{1}P_k^{(\alpha)} \rho P_k^{(\alpha)}-\rho\Bigg],
\end{equation}
where $P_k^{(\alpha)}:=|\psi_k^{(\alpha)}\>\<\psi_k^{(\alpha)}|$ are the rank-1 projectors on the MUB vectors of $\sigma_\alpha$. If we restrict our interest to the Hilbert spaces with the maximal number of $n+1$ MUBs (i.e. $n=s^r$ with prime $s$), then this generator can be easily generalized to higher dimensions,
\begin{equation}
\mathcal{L}(t)[\rho]=\sum_{\alpha=1}^{n+1}\gamma_\alpha(t)\Bigg[\sum_{k=0}^{n-1}P_k^{(\alpha)} \rho P_k^{(\alpha)}-\rho\Bigg]=-\frac 12\sum_{\alpha=1}^{n+1}\gamma_\alpha(t)\sum_{k=0}^{n-1}\left[P_k^{(\alpha)},\left[P_k^{(\alpha)},\rho\right]\right].
\end{equation}
The corresponding CPT map,
\begin{equation}
\Lambda(t)[\rho]=\pi_0(t)\rho+\frac{1}{n-1}\sum_{\alpha=1}^{n+1}\pi_\alpha(t)\Bigg[n\sum_{k=0}^{n-1}P_k^{(\alpha)} \rho P_k^{(\alpha)}-\rho\Bigg],
\end{equation}
is called the generalized Pauli channel and had already been analyzed in \cite{Nath,moje2}. Geometrically, the generator of the evolution reads
\begin{equation}\label{MUB_gen}
l(t)[p(t)]=\frac 18\sum_{\alpha=1}^{n+1}\gamma_\alpha(t)\sum_{k=0}^{n-1}X_{q_k^{(\alpha)}}^2p(t),
\end{equation}
where $q_k^{(\alpha)}=\<P_k^{(\alpha)}\>$. For each $\alpha$, $\square_\alpha:=\frac 14\sum_{k=0}^{n-1}X_{q_\alpha,k}^2$ denotes a differential operator (for $n=2$ this is the angular momentum \cite{moje}), and their sum corresponds to the Laplacian.
For $n=3$, the generators (\ref{MUB_gen}) and (\ref{geom_Weyl}) are equivalent. The same is true for (\ref{MUB_gen}) and (\ref{Pauli_L}) when $n=2^r$.

\section{CP and P-divisibility}

Recall that an evolution is called Markovian iff the corresponding dynamical map $\Lambda(t)$ is CP-divisible. This property manifests itself in the GKSL form of the generator $\mathcal{L}(t)$ (\ref{generator_general}), where $\gamma_k(t)\geq 0$. Observe that $\mathcal{L}(t)$ has the GKSL form for every generalization that we analyzed. Therefore, each of these evolutions is Markovian iff the decoherence rates are non-negative at any given time. It is worth noting that for $n=2$ \cite{moje} the evolution is Markovian iff the generalized Laplacian $\Delta_{\gamma(t)}$ is elliptic.

Now, let us relax the Markovianity condition and focus on satisfying the weaker requirement (\ref{P-div}) instead. It turns out that the necessary conditions for the maps solving (\ref{L_GM}) with $\gamma_{k_1k_2}^S(t)+\gamma_{k_1k_2}^A(t)=\gamma(t)$, (\ref{L_Weyl}) with $\gamma_{k_1k_2}(t)=\gamma_{n-k_1,n-k_2}(t)$, (\ref{L_Pauli}), and (\ref{L_gen_pauli}) to be P-divisible are that the eigenvalues of the respective channels are non-positive. One obtains this result by taking for $X$ the corresponding eigenvectors, which are simply the noise operators in the first three cases and $Q_k^{(\alpha)}:=P_k^{(\alpha)}-P_{k-1}^{(\alpha)}$ in the last one. Interestingly, this is equivalent to the requirement that the generalized Laplacian $\Delta_{\gamma(t)}$ is a non-positively defined operator.

%
%

\section{Conclusions}

The following paper describes the decoherence of a qudit within the geometrical formulation of quantum mechanics. We show that the Markovian master equation for density operators is equivalent to the Fokker-Planck equation for quantum probability distributions.

It turns out that there are at least four natural generalizations of the time-local generator for the qubit case, which differ in the choice of the noise operators. We showed that whenever the geometrical representation of the generator $l(t)=\<L(t)\>$ has the following form,
\begin{equation}\label{delta}
l(t)[p(t)]\propto\Delta_{\gamma(t)}p(t),
\end{equation}
one can associate the Markovianity of the evolution with the properties of the generalized Laplacian. 
If the evolution is non-Markovian but the corresponding map is P-divisible, then $\Delta_{\gamma(t)}$ is a non-positive operator.
Moreover, there are certain clues pointing into the direction that the evolution generated by (\ref{delta}) is Markovian (CP-divisible) iff $\Delta_{\gamma(t)}$ is elliptic. Unfortunately, this is hard to check for $n>2$.

The next step would be to try and associate the differential operator form the purely quantum part (\ref{PQ}) with a known geometrical quantity, as we did here for the classical-quantum part (\ref{QC}). It would be interesting to analyze its properties and find some connections to the divisibility property of quantum channels.

\section*{Acknowledgements}

This work was supported by the National Science Center project 2015/17/B/ST2/02026. The author thanks prof. Dariusz Chru\'sci\'nski for valuable remarks.

\end{document}